\documentclass[runningheads,a4paper]{llncs}

\usepackage{amssymb}
\usepackage{graphicx}

\usepackage{url}
\urldef{\mailsa}\path|elipachev@gmail.com|
\newcommand{\keywords}[1]{\par\addvspace\baselineskip
\noindent\keywordname\enspace\ignorespaces#1}

\begin{document}

\mainmatter  

\title{Boundary Value Problems for \\the Helmholtz Equation \\ for a Half-plane \\with a Lipschitz Inclusion}
\titlerunning{Boundary Value Problems for the Helmholtz Equation}
\author{Evgeny Lipachev}
\authorrunning{E.Lipachev}
\institute{ Kazan Federal University, 35 Kremlyovskaya ul. \\ Kazan, Russian Federation,\\
 \mailsa\\
\url{http://kpfu.ru/}}

\toctitle{Boundary value problems for the Helmholtz equation}
\tocauthor{Lipachev}

\maketitle

\begin{abstract}
This paper considers to the problems of diffraction of electromagnetic waves on a
half-plane, which has a finite inclusion in the form of a Lipschitz
curve. The diffraction problem formulated as boundary value problem for
Helmholtz equations and boundary conditions Dirichlet or Neumann on the boundary,
as well as the radiation conditions at infinity.
We carry out research on these problems in generalized Sobolev spaces.
We use the operators of potential type, that by their properties are analogs of the
classical potentials of single and double layers.
We proved the solvability of the boundary value problems of Dirichlet and
Neumann. We have obtained solutions of boundary value problems in the
form of operators of potential type.  Boundary problems
are reduced to integral equations of the second kind.
\end{abstract}

\section{Introduction}\label{intro}
At present, boundary problems on Lipschitz domains have been formed
in a special field of research. This is caused both with
applications in electrodynamics and other fields, and with the
theoretical importance of these studies.  An extensive bibliography
on this range of issues and the most important the results are given
in the  articles M.S.~Agranovich~\cite{agranovich-2015,agr2002,agr-men},
M.~Costabel \cite{costabel}, B.~Dahlberg and C.~Kenig~\cite{dahlberg},
D.~Jerison and C.~Kenig \cite{je-ke1981,je-ke-1981,je-ke-1995},
M.~Mitrea and M.~Taylor \cite{mitr-Ta1999,mitr-Ta2000}.
A detailed summary of the main results of the theory of boundary value problems
on Lipschitz domains can be found in~\cite{agranovich-2015}, \cite{grisvard},
\cite{McL00}.
An essential part of this theory is devoted to boundary integral operators.
A summary of the theory of integral operators on Lipschitz Domains
is contained in article~\cite{costabel}. Note that this work is the most cited in this field.
An important property of the Lipschitz domain is the possibility
approximations by infinitely differentiable domains from either side
boundaries of this domain (see \cite{agranovich-2015}, \cite{verchota}).
Lipschitz domains also satisfy the condition of a uniform cone (see
\cite{agranovich-2015,stein}).
These properties are used in the present paper to investigate the solvability
of boundary value problems of diffraction on an unbounded boundary with a Lipschitz inclusion.

In the papers \cite{lip-vuz:2001}--\cite{lip-uczap:2006} we investigated the boundary value problems
for the Helmholtz equation in domains with an rough smooth and piecewise smooth boundaries.
These studies are based on the use of generalized potentials of single and double layers.
In contrast to the classical potentials defined on closed domains, generalized potentials
are considered on open curves and on domains with an infinite boundary.
In the present paper this technique is extended to the case of a half-plane with a Lipschitz inclusion of finite size.

In the study of boundary value problems on Lipschitz boundaries we are introduced the operators of potential type.
These operators are analogues analogues of operators of single and double layers and
have properties close to those of the classical potentials of single and double layers,
which makes it possible to apply, after necessary refinements, the same reasoning technique
as in the classical case~(e.g., see~\cite{agr-men}, \cite{costabel}).

\section{Statement of the problem}\label{S:2}
A function $h: \mathbb{R} \rightarrow \mathbb{R}$ is called Lipschitz if there exist $C>0$ such that
$\left|h(x) - h(x')\right| \leqslant C \left|x-x'\right|$ for all $x, x' \in \mathbb{R}$.

Let $D \equiv D^+ =\left\{ (x,y) \in \mathbb{R}^2: y>h(x)\right\}$, where $h (x)$ is Lipschitz function with finite support.
In the terminology of publications~\cite{agr2002,agr-men,dyn1986} this domain is a special Lipschitz domain.
In~\cite{dahlberg} recommender use the notation $\Omega$ for bounded Lipschitz domains and $D$ for special Lipschitz domains respectively.

The boundary of the domain $D$ can be represented in the form $\partial D = \Gamma_1 \cup \Gamma_2$
\[
 \Gamma_1=\left\{(x,h(x)): x \not\in [0,d]\right\}, \Gamma_2=\left\{(x,0): x \in [0,d]\right\},
\mathrm{supp}\ h =[0,d].
\]
We note that this is a special case of rough boundary
(e.g., see \cite{lip-vuz:2006}, \cite{maradudin}, \cite{lip-tr:2011}). Let $D^- =\left\{ (x,y) \in \mathbb{R}^2: y<h(x)\right\}$.

We formulate the problem in the following terms. Find a function
$u(x,y) \in H^1 (\Omega)$, such that
\begin{equation}\label{Ed:e1}
\Delta u (M) + k^2 u (M) = 0, \quad M=(x,y) \in D,
\end{equation}
and the Dirichlet boundary condition
\begin{equation}\label{Ed:e2}
\gamma u (P)=  f(P), \quad P \in \partial D,
\end{equation}
or the Neumann boundary condition
\begin{equation}\label{Ed:e3}
\gamma' u (P) =  g (P), \quad P \in \partial D,
\end{equation}
and the radiation condition at infinity
\begin{equation}\label{Ed:e4}
{u}^\ast = e^{i k r} O\left(\frac{1}{\sqrt {r}} \right),
\quad \frac{\partial {u}^\ast}{\partial r} - i k
{u}^\ast = e^{i k r} o\left( \frac{1}{\sqrt{r}} \right), \quad
r \rightarrow \infty,
\end{equation}
where $r = \sqrt{x^2 + y^2} $  and $ {u}^\ast (x,y) = u (x,y) -
\widetilde{u} (x,y)$.
Through $\widetilde{u}^+$, $\widetilde{u}^-$ denoted by the solution of the
diffraction problem on the half-plane~(\cite{lip-vuz:2001,lip-vuz:2006,lip-uczap:2006}).

Here $f \in H^{1/2} (\partial D)$, $g
\in H^{-1/2} (\partial D)$, $k\in \mathbb{C}\setminus \{0\}$, $\mbox{\rm Im } k \geqslant 0$.
Through $H^t$ denoted the Sobolev spaces (e.g., see~\cite{agranovich-2015,tartar}).
Let $P\in \partial D$, denote by
\[
\gamma u (P) \equiv \left. u \right|_{\partial D} (P) = \lim_{M\rightarrow P, M\in \Lambda_\alpha (P)} u (M)
\]
-- trace of the function $u$ on the boundary $\partial D$, where
\[
Q_\alpha (P) \equiv \left\{M\in D : |M-P|<(1+\alpha)d(M,\partial D)\right\}, \quad \alpha>0
\]
-- Luzin sector.
Denote by
\[
\gamma' u (P) \equiv \left.\partial_n u\right|_{\partial D} (P) = \left(n(P), \left.(\nabla u)\right|_{\partial D} (P)\right).
\]

Note that
$\gamma : H^t (D) \rightarrow H^{t-1/2} (\partial D)$,
$\gamma' : H^t (D) \rightarrow H^{t-3/2} (\partial D)$  (e.g., see~\cite{dyn1986}, \cite{mitrea:2013}).

These problems are used as a mathematical model for finding the electromagnetic field resulting from the diffraction of an electromagnetic
plane wave in regions with an infinite rough boundary.
From a physical point of view, the boundary value problem with condition~(\ref{Ed:e2})  corresponds to the diffraction problem
$TE$--polarized electromagnetic wave, and the problem with the condition~(\ref{Ed:e3}) -- the problem of diffraction $TM$--polarized wave (e.g., see \cite{tsang}).

\section{Uniqueness of the solution of boundary value problems}\label{S:3}
 To derive Green's formulas, the Lipschitz domain is approximated by domains with smooth boundaries.
According to the results of article~\cite{verchota} (e.g., see~\cite{agr-men})  for a Lipschitz domain
$D$ there is a sequence of $C^\infty$ domains, $\overline{D}_j \subset D$,
and homeomorphisms, $\Lambda_j: \partial D \rightarrow \partial D_j$,
such that
$\sup_{P\in \partial D} \left|\Lambda_j (P) - P\right| \rightarrow 0$
as $j \rightarrow \infty$ and $\Lambda_j (P) \in Q_\alpha (P)$ for all $j$ and all $P\in \partial D$,
the normal vectors $n \left(\Lambda_j (P)\right)$ to $D_j$ coverge pointwise a.e. and in $L_2 \left(\partial D\right)$
to $n(P)$.
For a special Lipschitz domain $D$ from the section~\ref{S:2}  choose an approximating sequence of domains with condition
\[
\Gamma_1=
\left\{\left(x,0\right): x \in \mathbb{R} \setminus \left[0,d\right]\right\}
\subset
\partial D_j \cap \partial D.
\]

\begin{theorem}\label{Th:1}
If the condition $\mathrm{Im } k \, \geqslant \, 0$
the boundary value problems (\ref{Ed:e1}), (\ref{Ed:e2}), (\ref{Ed:e4}) have no more than one solution.
\end{theorem}

\begin{theorem}\label{Th:1'}
If the condition $\mathrm{Im } k \, \geqslant \, 0$ and $\mathrm{Re } k \neq 0$
the boundary value problems (\ref{Ed:e1}), (\ref{Ed:e3}), (\ref{Ed:e4}) have no more than one solution.
\end{theorem}

\begin{proof}
Let $\left\{D_j\right\}_{j\in \mathbb{N}}$ be a system of smooth domains approximating the domain $D$.
Let $u$, $v$  be two solutions of the boundary value problem and $w=u-v$.

On the part of the boundary $\Gamma_1$, because $\Gamma_1 \subset \partial D \cap \partial D_j$, we have
\[
\left. \gamma_D \right|_{\Gamma_1} w= \left. \gamma_{D_j}
\right|_{\Gamma_1} w = \left. w \right|_{\Gamma_1} =0,
\left. {\gamma'}_D \right|_{\Gamma_1} w= \left. {\gamma'}_{D_j}
\right|_{\Gamma_1} w = \left. \partial_\nu w \right|_{\Gamma_1} =0.
\]
Further, since
\[
\gamma_{D_j} u \rightarrow \gamma_D u, \quad \gamma_{D_j} v \rightarrow \gamma_D v,
\]
we have
\[
\left. w \right|_{\partial D_j}\equiv \gamma_{D_j} w = \gamma_{D_j} (u-v) \rightarrow 0, \quad j \rightarrow \infty.
\]

Similarly,
\[
\left. \partial_\nu w \right|_{\partial D_j}\equiv {\gamma'}_{D_j} w = {\gamma'}_{D_j} (u-v) \rightarrow 0, \quad j \rightarrow \infty.
\]

Consider functions $w_s \in C^{\infty}$ such that $w_s \rightarrow w$, $s \rightarrow \infty$.
Let $R>d$ be a real number and $S_R$ a circle of radius $R$.
In the bounded smooth region $D_{j,R}=D_j \cap S_R$ we apply the second Green's formula to the functions $w_s$ è $\overline{w}_s$:
\[
 \int\limits_{D_{j,R}} \left(w_s \, \Delta \overline{w}_s +
\overline{w}_s \, \Delta w_s\right)\, d\sigma  \: = \:
\int\limits_{\partial D_{j,R}} \left(w_s \, \partial_\nu \overline{w}_s
- \overline{w}_s \, \partial_\nu w_s\right)
  \, d\ell_P.
\]
In the limit to $R \rightarrow \infty$, we obtain
\[
\int\limits_{D_{j}} \left(w_s \, \Delta \overline{w}_s +
\overline{w}_s \, \Delta w_s\right)\, d\sigma  \: = \:
\int\limits_{\partial D_{j}} \left(w_s \, \partial_\nu \overline{w}_s -
\overline{w}_s \, \partial_\nu w_s\right)
  \, d\ell_P.
\]
Further, in the last formula we pass to the limit with respect to $s \rightarrow \infty$ and take into account the limit relations
\[
w_s (P) \rightarrow w (P), \quad \overline{w}_s (P) \rightarrow
\overline{w} (P), \quad P \in D_j,
\]
\[
\left. w_s\right|_{\partial D_j} \rightarrow \left.
w\right|_{\partial D_j}, \quad \left.
\overline{w}_s\right|_{\partial D_j} \rightarrow \left.
\overline{w}\right|_{\partial D_j}, \quad
\]
\[
\left. \partial_\nu w_s\right|_{\partial D_j} \rightarrow \left.
\partial_\nu w\right|_{\partial D_j}, \quad \left.
\partial_\nu \overline{w}_s\right|_{\partial D_j} \rightarrow \left.
\partial_\nu \overline{w}\right|_{\partial D_j}.
\]
As a result, we obtain
\[
\int\limits_{D_{j}} \left(w\, \Delta \overline{w} + \overline{w} \,
\Delta w\right)\, d\sigma   =  \int\limits_{\partial D_{j}}
\left(w \, \partial_\nu \overline{w} - \overline{w} \,
\partial_\nu w\right)
  \, d\ell_P.
\]
Now we pass to the limit with respect to $j \rightarrow \infty$, taking into account the approximation properties of a Lipschitz domain $D$ by smooth domains.
\begin{equation}\label{doked12}
\int\limits_{D} \left(w\, \Delta \overline{w} + \overline{w} \,
\Delta w\right)\, d\sigma   =  \int\limits_{\partial D}
\left(\gamma_D w \, {\gamma'}_D \overline{w} - \gamma_D\overline{w} \,
{\gamma'}_D w\right)
  \, d\ell_P.
\end{equation}

Because the
$\Delta w = -k^2 w, \quad \Delta \overline{w} = -\overline{k}^2 \overline{w}$,
then the left-hand side of the equality~(\ref{doked12}) takes the form
\[
i 4 \mathrm{Re }k\, \mathrm{Im }k \int_D \left|w\right|^2 \,
\partial \sigma.
\]

As a consequence of the boundary conditions of the boundary value problem, the right-hand side of equation~(\ref{doked12}) is $0$.

Since $\mathrm{Re } k \neq 0$, $\mathrm{Im } k \geq  0$, we have
\[
\int\limits_D \left|w\right|^2 \, d\sigma = 0.
\]

From the last relation we conclude that $w \equiv 0$ is satisfied in region $D$ and, as a consequence, we obtain $u=v$.
\end{proof}

\section{Existence of solutions of boundary value problems}\label{S:4}

One method of solving boundary value problems of diffraction is the method of integral equations (e.g., see \cite{colt-kr}, \cite{maz1}).
In the classical theory, boundary value problems are considered on bounded domains with a sufficiently smooth boundary,
the potentials of the single and double layers are used, as well as the technique of Green's formulas.
In the case of a deterioration of the properties of the boundary, it is required to refine the definitions of the potentials
and the conditions for the applicability of the Green's formulas.
In the study of boundary value problems for the Helmholtz equation on rough boundaries, we used generalized potentials (\cite{lip-vuz:2001}--\cite{lip-uczap:2006}).
In the case of Lipschitz boundaries, the extension of the concept of a potential is called a potential type operator.
We note that in the case of Lipschitz domains the properties of operators of potential type are analogous to those of classical
potentials of a single and double layers, in particular, formulas for the jump of values on the boundary.

As shown in my works~\cite{lip-vuz:2001}--\cite{lip-uczap:2006} in the case of an rough boundary from class $C^{(1,\nu)}$, $\nu \in (0,1]$,
under the conditions $\mathrm{Im}\, k \geqslant 0$, $\mathrm{Re}\, k \neq 0$
the boundary value problem has a unique solution and for the solution we have the representation
\[
u(x,y) = \widetilde{u} (x,y) + v(x,y),
\]
\begin{equation}\label{E:w}
v (x,y) = \left(W (k) \varphi \right) (x,y) = \int\limits_{\partial D}
\partial_{n(P)} G_1 (k;M,P) \varphi (\tau) \, ds_P
\end{equation}
-- in the case of the problem with the condition~(\ref{Ed:e2}) on the boundary $\partial D$ and
\begin{equation}\label{E:v}
v (x,y) = \left(V (k)\varphi\right) (x,y) = \int\limits_{\partial D} G_2
(k;M,P) \varphi (\tau) \, ds_P
\end{equation}
-- in the case of condition~(\ref{Ed:e3}) on the boundary. The function$\varphi (x)$ is a solution of the integral equation
\begin{equation}\label{E:ie1}
-\pi \varphi (x) + \int\limits_{0}^d \partial_{n(P)} G_1 (k;M,P)
\sqrt{1+{h'}^2 (\tau)} \, \varphi (\tau) \, d\tau = - f(M) +\widetilde{u}
(M)
\end{equation}
-- in the case of condition~(\ref{Ed:e2}) on the boundary and
\begin{equation}\label{E:ie2}
-\pi \varphi (x) + \int\limits_{0}^d \partial_{n(M)} G_2 (k;M,P)
\sqrt{1+{h'}^2 (\tau)} \, \varphi (\tau) \, d\tau = - g(M)+\partial_{n(M)}\widetilde{u} (M)
\end{equation}
-- in the case of condition~(\ref{Ed:e3}) on the boundary.

These formulas use functions
\[
G_m (k;M,P) = \frac{\pi i}{2}
 \left\{
   H_0^{(1)} (k r) + (-1)^m H_0^{(1)} \left(k r^\ast\right)
 \right\},
\quad m=1,2,
\]
$M = (x,h(x))$, $P=(\tau, h(\tau))$, $r=\sqrt{(x-\tau)^2
+(y-h(\tau))^2}$, $r^\ast=\sqrt{(x-\tau)^2 +(y+h(\tau))^2}$,
$d=\mathrm{supp} \, h(x)$ -- the length of the irregular part of the boundary $\partial D$.
Through $H_0^{(1)} (z)$ denotes the Hankel function of the first kind of order zero (e.g., see~\cite{abramowitz}).

In the case of a special Lipschitz domain $D$, we consider the operators
\begin{equation}\label{pot5}
\left(\mathcal{V} (k) \varphi \right) (M) = \int\limits_{\Gamma_2}
G_2 (k;M,P) \varphi (\tau) \, d\ell_P, \quad   M \in D,
\end{equation}
\begin{equation}\label{pot6}
\left(\mathcal{W} (k) \psi \right)(M) = \int\limits_{\Gamma_2}
\partial_{n (P)} G_1 (k;M,P) \psi (\tau) \, d\ell_P, \quad  M \not\in
\partial D.
\end{equation}

Here $n (P)$ is the unit normal vector at the point $P$ directed to the region $y>0$.
This vector is defined for almost all $P \in \partial D$.

For $M=\left(x,h(x)\right)\in \partial D$ we define
\[
V (k) \varphi (x) = \lim_{\varepsilon \to
0}\int\limits_{|M-P|>\varepsilon} G_2 (k;M,P) \varphi (\tau) \,
d\ell_P.
\]

From the results of \cite{verchota} (e.g., see \cite{maz1}) the following statements.

\begin{lemma}\label{Th:6-1}
If $\varphi \in L_p \left(\partial D\right)$, $1<p<\infty$, then there is the direct value of the normal derivative of the operator~(\ref{pot5})
\[
V' (k) \varphi (x) \equiv \displaystyle\left[\partial_{n (M)}
\left(\mathcal{V} (k) \varphi \right) \right] (x) =
\displaystyle\lim_{\varepsilon \to
0}\displaystyle\int\limits_{|M-P|>\varepsilon}
\partial_{n (M)} G_2 (k;M,P) \varphi (\tau) \, d\ell_P.
\]
Here the limit is understood in the sense of convergence in $L_p\left(\partial D\right)$ or pointwise convergence for almost all $M \in \partial D$.

The normal derivative $\partial_{n (M)} \left(\mathcal{V} (k) \varphi \right)$
for almost all $M \in \partial D$ has nontangential limits $\partial_{n (M)} \left(\mathcal{V} (k) \varphi \right)_{\pm}$
on the side $D^{\pm}$, which are expressed by formulas
\begin{equation}\label{potprsl1}
\partial_{n (M)}
\left(\mathcal{V} (k) \varphi \right)_{\pm} = \pm \frac{1}{2}
\varphi + V' (k) \varphi.
\end{equation}
\end{lemma}

\begin{lemma}\label{Th:6-2}
If $\psi \in L_p \left(\partial D\right)$, $1<p<\infty$, then almost everywhere on $\partial D$ there exists the limit
\begin{equation}\label{pot9}
W (k) \psi (x) = \lim_{\varepsilon \to
0}\int\limits_{|M-P|>\varepsilon} \partial_{n (P)} G_1 (k;M,P)
\psi (\tau) \, d\ell_P.
\end{equation}
\end{lemma}

This limit is called the direct value of the operator $\mathcal{W} (k)\psi$.

\begin{lemma}\label{Th:6-3}
If $\psi \in L_p \left(\partial D\right)$, $1<p<\infty$, then almost everywhere on the boundary $\partial D$ has nontangential limits
$\left(\mathcal{W} (k) \psi \right)_{\pm}$ on the side $D^{\pm}$, and the following equalities hold:
\begin{equation}\label{potdvsl}
\left(\mathcal{W} (k) \psi \right)_{\pm}\psi = \mp \frac{1}{2}
\psi + W (k) \psi.
\end{equation}
\end{lemma}

For the operators~(\ref{pot5}), (\ref{pot6}), the basic potentials are satisfied (e.g., see \cite{agranovich-2015,agr2002,agr-men}).
Therefore, we can consider them analogues of the potentials of a single and a double layer.
These operators are called operators of the potential type.

The following propositions hold (e.g., see \cite{costabel,je-ke1981,je-ke-1981,je-ke-1995,mitr-Ta1999})
\[
\mathcal{V} (k) : H^{t-1/2} \left(\partial D\right) \rightarrow
H^{t+1/2} \left(\partial D\right), \quad
\]
\[
\mathcal{W} (k) : H^{t+1/2} \left(\partial D\right) \rightarrow
H^{t+1/2} \left(\partial D\right), \quad -\frac{1}{2} \leqslant t \leqslant
\frac{1}{2},
\]
\[
\mathcal{V} (k) : L_2 \left(\partial D\right) \rightarrow H^{1}
\left(\partial D\right),
\quad
\mathcal{W} (k) : H^{1} \left(\partial D\right) \rightarrow H^{1}
\left(\partial D\right).
\]

Let $\left\{D_j\right\}_{j\in \mathbb{N}}$ be a system of smooth domains approximating the Lipschitz domain $D$.
We denote by $\left\{u_j\right\}$ the sequence of solutions of boundary value problems in smooth domains $D_j$.
This sequence can be adjusted so that
$\left.u_j\right|_{D_k} = u_k, \quad  k\leqslant j$.

In each smooth domain $D_j$ we consider the boundary value problem
\[
\Delta u (M)+ k^2 u (M) =0, \quad M \in D_j,
\]
\[
\left. u\right|_{\partial D_j} = \left. u_{j+1}\right|_{\partial
D_j}
\]
and, in addition, we require that the function $u$ satisfy the radiation conditions~(\ref{Ed:e4}).

Note that the value of the function $u_{j+1}$  on $\partial D_j$ is defined by virtue of the fact that $D_j \subset D_{j+1}$.

For a given boundary value problem, the solvability conditions are satisfied; therefore, for each $j$ there exists a classical solution $\dot u_j$ of this problem.

But, as a function of $u_{j+1}$ in $D_j$ also satisfies the conditions of the boundary value problem, then, by the uniqueness, we get
\[
\dot u_j = \left.u_{j+1}\right|_{D_j}.
\]
Note also that on the general section boundary areas $D_j$ and $D_{j+1}$ we have
\[
\left. \dot u_j\right|_{\partial D_j \cap \partial D_{j+1}} =
\left.u_{j+1}\right|_{\partial D_j \cap \partial D_{j+1}} =
\left.f\right|_{\partial D_j \cap \partial D_{j+1}}
\]
-- in the case of the Dirichlet problem and
\[
\left. \partial_\nu \dot u_j\right|_{\partial D_j \cap \partial
D_{j+1}} = \left. \partial_\nu u_{j+1}\right|_{\partial D_j \cap
\partial D_{j+1}} = \left.g\right|_{\partial D_j \cap \partial
D_{j+1}}
\]
-- in the case of the Neumann problem.

For each finite set of solutions $\left\{u_0,u_1, \ldots,
u_j\right\}$ of the boundary value problems in the domains $D_0, D_1, \ldots,
D_j$ can perform such adjustment, starting from the value $j$ and reducing to index $0$.
As a result, reasoning by induction, we obtain the sequence of functions $u_0 (x,y), u_1 (x,y), \ldots$ satisfying the following conditions.

(i) For each $j\in \mathbb{N}$, function $u_j (x,y)$ at all points of the domain $D_j$ is a solution of the Helmholtz equation.

(ii) If $k\leqslant j$, then $u_k (x,y) = \left. u_j (x,y)\right|_{D_k}$.

(iii) In the case of the Dirichlet problem,
$\left. u_j \right|_{\partial D_j \cap \partial D_{j+1}} = \left. f\right|_{\partial D_j \cap \partial D_{j+1}},$
and
$\left. \partial_\nu u_j \right|_{\partial D_j \cap \partial D_{j+1}} = \left. g\right|_{\partial D_j \cap \partial D_{j+1}}$
-- in the case of the Neumann problem.

(iv) In the domain $D_j$, the function $u_j$ can be represented as a generalized potential with a density $\varphi_j (x)$ found as a solution of the integral equation (\ref{E:ie1}) or (\ref{E:ie2}).

Thus, we obtain a sequence of functions $\varphi_j (x)$.
Let us show that this sequence is fundamental in $L_2 \left[0,d\right]$.
Let $R>d$ be a real number, we define the domain
$S_R =\displaystyle\left\{(x,y): x^2 + y^2 \geqslant R^2, \ \ y>0\right\}.$

Consider in $S_R$ the two functions $u_i$ and $u_j$ of this sequence and for definiteness, assume that $i>j$.

As shown, the functions $u_i$ and $u_j$ coincide in domain $D_j$, and hence in $S_{R,i}=S_R \cap D_i$ we have
\begin{equation}\label{equmin}
u_i (x,y) - u_j (x,y) =0, \quad i>j.
\end{equation}

We write the last relation, using the representation of solutions of boundary value problems in the form of potentials (\ref{E:w}) and (\ref{E:v}).

In the case of the Dirichlet problem, we have
\[
\displaystyle\left. \left(u_i  -  u_j\right)\right|_{S_{R,i}}
 = \displaystyle\int\limits_{\partial D_i \setminus \Gamma_1}
\partial_{n (P)} G_1(k;M,P) \varphi_i (\tau) \, d\ell_P -
\displaystyle\int\limits_{\partial D_j
\setminus \Gamma_1}
\partial_{n (P)} G_1(k;M,P) \varphi_j (\tau) \, d\ell_P.
\]

In the case of the Neumann problem, we have
\[
\displaystyle\left. \left(u_i  -  u_j\right)\right|_{S_{R,i}}
 = \displaystyle \int\limits_{\partial D_i \setminus \Gamma_1}
G_2 (k;M,P) \varphi_i (\tau) \, d\ell_P -
\displaystyle \int\limits_{\partial D_j
\setminus \Gamma_1} G_2 (k;M,P) \varphi_j (\tau) \, d\ell_P.
\]

Since the function $\varphi_i$ is defined as in domain $D_i$ and in the domain $D_j$ ($D_j \subset D_i$ with $j<i$),
we can consider the potential on the boundary $\partial D_j$ with a density $\varphi_i$.
Then the following relations hold:
\[
\int\limits_{\partial D_i \setminus \Gamma_1}
\partial_{n (P)} G_1(k;M,P) \varphi_i (\tau) \, d\ell_P -
\int\limits_{\partial D_j \setminus \Gamma_1}
\partial_{n (P)} G_1(k;M,P) \varphi_j (\tau) \, d\ell_P =
\]
\[
=\int\limits_{\partial D_i \setminus \Gamma_1}
\partial_{n (P)} G_1(k;M,P) \varphi_i (\tau) \, d\ell_P -
\int\limits_{\partial D_j \setminus \Gamma_1}
\partial_{n (P)} G_1(k;M,P) \varphi_i (\tau) \, d\ell_P 
\]
\[
+ \int\limits_{\partial D_j \setminus \Gamma_1}
\partial_{n (P)} G_1(k;M,P) \left(\varphi_i (\tau)  - \varphi_j (\tau)\right)\,
d\ell_P.
\]

Hence, from the last relation and from (\ref{equmin}) we obtain
\[
\displaystyle\int\limits_{\partial D_j \setminus \Gamma_1}
\partial_{n (P)} G_1(k;M,P) \left(\varphi_i (\tau)  - \varphi_j (\tau)\right)\,
d\ell_P  =
 \phantom{AAAAAAAAAAA}
\]
\[
= \displaystyle\int\limits_{\partial D_j \setminus \Gamma_1}
\partial_{n (P)} G_1(k;M,P) \varphi_i (\tau) \, d\ell_P -
\int\limits_{\partial D_i \setminus \Gamma_1}
\partial_{n (P)} G_1(k;M,P) \varphi_i (\tau) \, d\ell_P.
\]

Similarly, in the case of the Neumann problem, we have
\[
\displaystyle\int\limits_{\partial D_j \setminus \Gamma_1}
G_2(k;M,P) \left(\varphi_i (\tau)  - \varphi_j (\tau)\right)\,
d\ell_P  =
 \displaystyle\int\limits_{\partial D_j \setminus \Gamma_1}  G_2(k;M,P)
\varphi_i (\tau) \, d\ell_P
\]
\[
- \displaystyle\int\limits_{\partial
D_i \setminus \Gamma_1}  G_2(k;M,P) \varphi_i (\tau) \, d\ell_P.
\]

The right-hand sides of the last relations for $j \rightarrow \infty$ tend to $0$, since $\partial D_j$, $\partial D_i$ approach $\partial D$.
From these relations and the properties of approximation by smooth domains, convergence of the sequence of functions $\left\{\varphi_i\right\}$ follows.
We denote by $\psi^\ast$, $\varphi^\ast$ the limits of the sequence of densities in the case of the Dirichlet and Neumann problem, respectively.

We show that the functions $u^\ast = \mathcal{W} (k) \psi^\ast + \widetilde{u}$, $v^\ast = \mathcal{V} (k) \varphi^\ast + \widetilde{v}$
satisfy the conditions of the boundary value problem.

The use of the trace operator $\gamma$ to the function $u^\ast$ and the operator $\gamma'$ to the function $v^\ast$,
by Lemmas~\ref{Th:6-1} and \ref{Th:6-3}, leads to the relations
\[
\gamma u^\ast = -\frac{1}{2} \psi + W(k) \psi +
\left.\widetilde{u}\right|_{\partial D},
\quad
\gamma' u^\ast = \frac{1}{2} \varphi + V'(k) \varphi +
\left.\widetilde{v}\right|_{\partial D}.
\]
The latter relations are understood in the sense of the nontangential limit of functions whose values on $\partial D$
have the same values as $f$ or $g$ (depending on the boundary condition).
Therefore, we arrive at the conclusion that the relations
\[
\gamma u^\ast (P)  =  f (P), \quad
\gamma' v^\ast (P)  =  g (P), \quad P \in \partial D.
\]

As a result, we conclude that the following theorems hold.
\begin{theorem}\label{Th:7-1}
Under condition $\mbox{\rm Im } k \geqslant 0$, the sequence $\displaystyle \left\{\psi_j (x)\right\}_1^\infty$ of the solutions of
integral equations~(\ref{E:ie1}) converges in a space $L_2 [0,d]$ to a function $\psi (x)$ such that the function
\[
u (M)  = \widetilde{u}\, (M) + \left(\mathcal{W} (k) \psi
\right)(M)
\]
is a solution of the boundary value problem with the Dirichlet condition on the boundary.
We denote by $\widetilde{u}$ the solution of the  Dirichlet problem on the half-plane, and $\mathcal{W} (k) \psi$
is the potential type operator defined by~(\ref{pot6}).
\end{theorem}

\begin{theorem}\label{Th:7-2}
Under conditions $\mbox{\rm Im } k \geqslant 0$ and $\mbox{\rm Re } k \neq 0$, the sequence $\displaystyle \left\{\varphi_j (x)\right\}_1^\infty$
of the solutions of the integral equations~(\ref{E:ie2}) converges to a function $\varphi (x)$ in the space $L_2 [0,d]$ such that
the function
\[
u\, (M) = \widetilde{v}\, (M) + \left(\mathcal{V} (k) \varphi
\right) (M)
\]
is a solution of the boundary value problem with the Neumann condition on the boundary.
We denote by $\widetilde{v}$ the solution of the Neumann boundary value problem on the half-plane, and $\mathcal{V}(k) \varphi$
is the potential type operator defined by~(\ref{pot5}).
\end{theorem}

\begin{theorem}\label{Th:7-2}
There exists a unique solution $u(x,y)$ of the boundary value problems under consideration, and the representations hold
\[
u = \widetilde{u} + \mathcal{W} (k) \left[\left(I-W(k)\right)^{-1}
f\right] \quad \mbox{\it in the case of the Dirichlet problem},
\]
\[
u = \widetilde{v} + \mathcal{V} (k)
\left[\left(I-V'(k)\right)^{-1} g\right] \quad \mbox{\it in the case of the Neumann problem}.
\]
\end{theorem}

\end{document}